\documentclass{article}

\def\theoremstyle{\bf}
\newtheorem{theorem}{\theoremstyle Theorem}
\newtheorem{Prop}[theorem]{\theoremstyle Proposition}
\newenvironment{proof}[1][{}]{{\theoremstyle Proof#1. }\nopagebreak}{}

\usepackage{german}
\usepackage{amssymb}
\usepackage{exscale}
\usepackage{array}
\usepackage[round]{natbib}
\usepackage{bm} 
\usepackage{graphicx}

\def\myaddress{Institut f\"ur Mathematik,
  Universit\"at Mannheim\\A5, 6, D-68131 Mannheim 
}
\def\myemail{schlather@math.uni-mannheim.de}

\usepackage{moreverb}

\usepackage[dvips,colorlinks,bookmarksopen,bookmarksnumbered,citecolor=red,urlcolor=red]{hyperref}

\usepackage{natbib}

\newcommand\BibTeX{{\rmfamily B\kern-.05em \textsc{i\kern-.025em b}\kern-.08em
T\kern-.1667em\lower.7ex\hbox{E}\kern-.125emX}}

\def\RR{{\mathbb R}}
\def\NN{{\mathbb N}}

\begin{document}

\title{A parametric variogram model
  bridging between stationary and intrinsically stationary processes}

\author{Martin Schlather\\\myaddress\\\myemail}

\maketitle

\emph{Abstract}:
A simple  variogram model with two parameters is presented that includes
the  power variogram for the fractional Brownian motion, a modified 
De Wijsian model,
 the generalized Cauchy model 
and the multiquadrics model.
One parameter controls the smoothness of the process.
The other parameter allows for a smooth parametrization between
stationary and intrinsically stationary second order processes
in a Gaussian framework,
or between mixing and non-ergodic max-stable processes
when modeling spatial extremes by a Brown-Resnick process.

\medskip
{\em Keywords: }{Brown-Resnick process; Cauchy model;  fractional Brownian motion;
Gaussian process; variogram}

\section{Introduction}
 
A Gaussian random field $Z$ is a popular process to model
spatial data in $\RR^d$, mostly assuming some kind of at least underlying
 stationarity and isotropy
of the field.
A weakly
stationary random field can be characterized by a
 bounded variogram $\gamma$, where
$\gamma(h) = \frac12 E \{Z(h) - Z(0)\}^2$, $h\in\RR^d$ \citep{gneitingguttorp10}.
In contrast, the variogram of an intrinsically stationary random field
might be unbounded.
Although any unbounded variogram can be approximated arbitrarily closely on
any compact set by a bounded variogram \citep{schlathergneiting06},
 models that bridge between bounded and unbounded models
 seem to be used only rarely in statistics if at all, as appealing
 models are missing.

The Brown-Resnick process on $\RR^d$
\citep{brownresnick77,KSH09} is a max-stable
random field that 
is based on intrinsically stationary Gaussian random fields and
is used to model spatial extreme values,
such as heavy rainfall \citep{fawcett2013estimating,thibaud2013threshold},
 avalanches \citep{blanchetdavison11},
extreme temperature \citep{davison2013geostatistics} or
extreme wind \citep{
zhang2014topography}.
If the underlying variogram model
is bounded, then the corresponding Brown-Resnick field is not ergodic
whilst a one-dimensional process has been shown to be mixing when the variogram model
$\gamma(h)$, $h\in\RR$, grows faster that $4\log |h|$ \citep{wangstoev10,kabluchkoschlather10}.
Ergodicity might be assumed if the extremes are caused by local
events, e.g., storm events, while 
non-ergodic models are suitable for spatially extended events such 
as cyclonic rainfall.
When the kind of 
causative process for spatial extremes
is not clear, a parametric model that allows for both
ergodic and non-ergodic fields is advantageous.

For lack of awareness of suitable parametric models that bridge between
stationary and intrinsically stationary Gaussian random fields, studies in
both traditional geostatistics and spatial extreme value analysis
using the Brown-Resnick model include two classes of models
to cover the cases of bounded and unbounded variograms models,
cf.\ \cite{thibaud2013threshold} or \cite{wadsworth2014efficient}, for
instance.

Here, we present an easily accessible, rotation invariant variogram
model with
two parameters that bridges
between the class of bounded variograms and the class of unbounded ones. It
includes several well-known models, but contains also novel parameter combinations.
Hence, the model might be particularly useful in any
application of spatial statistics.

\section{The model}
For $0<\alpha\le 2$ and $-\infty < \beta\le 2$ 
let the function $\gamma_{\alpha,\beta} : \RR^d
\rightarrow [0,\infty)$ be defined as 
  $$
  \gamma_{\alpha,\beta}(h) = \frac{(1 + \|h\|^\alpha)^{\beta/\alpha} -
    1}{2^{\beta/\alpha} - 1}, 
  $$
where the limiting function is 
${\log(1 +  \|h\|^\alpha)}/{\log 2}$ as $\beta \rightarrow 0$.

\begin{Prop}
The function $\gamma_{\alpha,\beta}$ is a variogram in $\RR^d$ for any 
$d\in\NN$, $0<\alpha\le 2$ and $-\infty < \beta\le 2$.  
\end{Prop}

\begin{proof}
  A positive constant and the identity  being complete 
  Bernstein functions,  Corollary 7.12 and Proposition 7.10
  in  \cite{SSV} yield
  that $f(\lambda) = (1 + \lambda^\delta)^{\varepsilon/\delta}$ 
  is a Bernstein function for any $0 < \varepsilon \le 1$
  and $0 < |\delta| \le 1 $.
  Now $f(g)$ is a negative definite function for any negative
  definite function $g$ and any Bernstein function $f$ \citep[Chapter 4]{SSV}.
  Choosing $g(h)=\|h\|^2$, $\varepsilon=\beta/2$ and $\delta
  =\alpha/2$, we get the assertion for positive values of $\beta$. 
  For negative values of $\beta$ we refer to the fact that
  $(1+\gamma)^\beta$ is a positive definite function for any variogram
  $\gamma$, see \citep[p.~36]{schlather12}, for instance.
  By the basic properties of negative definite functions the limiting
  cases obtained by point-wise convergence are included.

\end{proof}

The parameter $\alpha$ models the smoothness of both the variogram 
and
the corresponding (Gaussian) random field \citep{adler,scheuerer10} whereas
$\beta$ indicates the long range behavior \citep{gneiting.powerlaw}.
For negative values of $\beta$ the
model is bounded and heavy tailed, whereas for $0 \le \beta \le 2$
the variogram $\gamma_{\alpha,\beta}$ is unbounded.
The normalizing factor is chosen such that $\gamma_{\alpha,\beta}(1)=1$.

The model simplifies for $\beta=\alpha$ 
to the power model of a fractional Brownian random field
\citep{chilesdelfiner}, 
for $0 < \beta \le \alpha$ to
a generalization of the fractional Brownian model described in 
 \cite{gneiting.nonseparable} and \cite{S10},
and for $\beta < 0$ to the 
 generalized Cauchy model \citep{gneiting.powerlaw}.
It also generalizes partially the multiquadrics and inverse multiquadrics used in 
approximation theory where $\alpha=2$
and $\beta/\alpha \in \RR \setminus \NN_0$, see, for instance,
 \cite{buhmann1990multivariate}, \cite{wendland} or \cite{linming06}. 
As  $\beta \rightarrow 0$, the limiting model 
equals a modified version of the
 De Wijsian model \citep{wackernagel,matheron62}. 
The case $0<\alpha<\beta\le2$ is novel.

The presented variogram model can be generalized. 
Indeed, the proof of the proposition shows
 that $\gamma_{\delta, g} = \{(1+g)^\delta -1\} /
(2^\delta -1)$ is a variogram in $\RR^d$ for any variogram $g$
and $-\infty < \delta \le 1$ with limiting case $\gamma_{0, g} =\log(1 + g)/\log(2)$.

\bibliographystyle{plainnat}

\end{document}